%% file: main.tex
\documentclass[12pt]{article}
\input{preamble}

\date{}
\title{Improved Speed via Regional Fulfillment\thanks{This material is based upon work supported in part by the Air Force Office of Scientific Research under award number FA9550-23-1-0031. D.~Hathcock supported by NSF Graduate Research Fellowship grant DGE-2140739.
}}

\author{Daniel Hathcock\thanks{Corresponding Author: {\tt dhathcoc@alumni.cmu.edu}, Carnegie Mellon University, Pittsburgh, PA, USA.}
\and 
R. Ravi\thanks{{\tt ravi@andrew.cmu.edu}, Ravi holds concurrent appointments as a Professor of Operations Research at Carnegie Mellon University and as an Amazon Scholar. This paper describes work performed at CMU and is not associated with Amazon.}
\and
Amitabh Sinha\thanks{{\tt amitabsi@amazon.com}, Amazon.com}
}

\begin{document}

\maketitle

\begin{abstract}
In e-retail, order fulfillment speed has become one of the most important metrics affecting customer satisfaction. While common wisdom dictates that maintaining a large global fulfillment network maximizes efficiency via economies of scale, recent evidence has shown that breaking up the network into smaller regions can yield significant speed improvements. In this paper, we consider a simple abstract model of order fulfillment by which we explain this phenomenon. We characterize fulfillment assignments satisfying an equilibrium condition based on the greedy fulfillment strategy, and quantify how the resulting fulfillment delay can be decreased by regionalizing the network. Finally, we provide some algorithmic results for computing low delay assignments, and some simulations supporting our equilibrium framework. 
\end{abstract}


\pagenumbering{gobble}

\pagenumbering{arabic}
\setcounter{page}{1}

\section{Introduction}

Winning the competition in e-retail is predicated on achieving a superior speed of order fulfillment. This has led to a variety of offerings such as `lightning' fulfillment (in under half hour~\cite{lightning}), sub--same-day, same-day and next-day fulfillment among e-retailers. In this setting, order management systems assign orders from customers to fulfillment centers (FCs) with the goal of meeting demand with as little delay as possible. A simple online strategy is to greedily assign an order to whichever FC can fulfill it most quickly - which is often the closest FC where the ordered product is in stock. The primary determinant of the time to fulfill an order is the distance to the customer 
so this requires FC locations and capacities that can support nearby fulfillment for all the different items that can be ordered. While the greedy strategy will usually fulfill orders using nearby FCs, imbalances in supply and demand may cause backlog to accumulate at some FCs, and in this case the greedy strategy will fulfill orders from FCs located farther away, decreasing the speed of fulfillment.

Recently, Amazon made a decision to 
regionalize their FC network~\cite{asn},
effectively splitting it into smaller geographic components, each of which independently runs the greedy strategy. At first glance, this appears surprising since the partitioning introduces more constraints in the fulfillment system. Indeed, the efficiency and cost savings of Amazon's network can be attributed to its global nature, taking advantage of economies of scale. However, if speed of fulfillment is a primary concern, we argue that regionalization can be used to improve the performance of the network in this respect. 
In this work, we present and analyze a simple model of regionalization that supports this viewpoint.

\paragraph{Backlog Equilibria:} We are given a set $\F$ of FC locations and a set $\D$ of demand locations with $\ell_{ij}$ representing the time it takes to satisfy the demand $i$ from FC $j$. 
To focus our model on the fulfillment speed aspect, we abstract away the complexities of fulfilling orders involving multiple items or different products, and consider the problem of fulfilling orders speedily for a single Stock Keeping Unit (SKU): Each FC $j$ has a capacity $C_j$ of units of this SKU that it can supply, and each demand location has demand $\Dem_i$ for this SKU. 
For simplicity, we assume $\sum_j C_j\geq \sum_i \Dem_i$. Moreover, if we assume that $\ell_{ij}$ is non-infinite for every pair, there is always a feasible matching of demand with supply from some FC. 

We are concerned with matchings of the demand to supply that follow the greedy strategy. Specifically, greedy fulfillments are matchings that assign demands only to FCs that incur the minimum delay for that demand. In general, this might result in a matching that exceeds the capacities for certain FCs. Therefore, we introduce the notion of \emph{backlogs} $\beta_j$ for each FC $j$, and define the delay of assigning $i$ to $j$ to be the time of fulfillment plus the backlog: $\ell_{ij} + \beta_j$. 
We think of the backlog as the amount of queued-up work at the FC represented in the same time units as the travel times between locations. Then the effective delay for fulfillment is the sum of the travel delay plus the accumulated backlog delay of processing at the FC.

The first question we address is whether there exist backlog values at the FCs, say $\beta_j \geq 0$ for all $j \in \F$, such that equilibrium is maintained and there is a feasible greedy assignment. Is there a simple characterization of these backlogs, and can they be found efficiently in arbitrary fulfillment networks? 

\paragraph{Regionalization:} To illustrate the application of understanding such backlog supported equilibria, consider the following example of a continuous demand instance on the line metric from \cite{sinha2026regionalize}. The line has unit length, and demand is uniformly generated at points in the line at the rate of one per unit time (day). The two FCs are located at 0 and 0.4 in the unit interval. The travel time for fulfillment is the distance to the FC. The fulfillment capacity of each FC is 0.5/day, thus the total capacity is exactly equal to the total demand. Since the left FC (FC1) at 0 is closer only to 0.2 of the demand, greedy assignment soon starts overloading the right FC (FC2). Eventually, in steady state, the right FC settles at a backlog of 0.2 units of demand, which with its processing rate of 0.5/day translates to 0.4 days of backlog. 
With this backlog, all demand in $[0,0.4)$ prefers FC1 to FC2 in terms of net delay. Also, all demand points in $[0.4,1]$ are indifferent between FC2 which has the backlog of 0.4 at location 0.4 or FC1 which has no backlog and at location 0. At equilibrium, this 0.6 measure of demand will be split as 0.1 to the left FC1 and 0.5 to FC2 so as to achieve load balance. However, this uniform proportional split leads to points in the whole interval $[0.4,1]$ to also be served by FC1 proportionally. The resulting average delay for fulfilling demand can be worked out to be 0.5.

On the other hand, suppose we split the interval into two segments $[0,0.5]$ and $[0.5,1]$ each served by FC1 and FC2 respectively, the average delay of fulfillment drops to 0.3. This illustrates the power of {\bf regionalization}: creating artificial regions within which we constrain greedy fulfillment. 

Using a discrete version of the simplified model above, we give partial answers to the questions: 
How much can regionalization decrease delay? Can good regionalizations be found efficiently? Along the way, we characterize backlog values that result in equilibria using the theory of bipartite matchings.

\subsection{Contributions}
Given an arbitrary fulfillment instance with at least as much supply as the demand, we give a characterization of backlog values that support greedy fulfillment equilibria where each demand is assigned to a supply node with minimum value of its travel time plus backlog, which also leads to a simple polynomial time algorithm to compute such values. Moreover, we show that any minimum travel time assignment in the original metric can be augmented with backlog values at the FCs to support an equilibrium solution, where all nodes are served by an FC with the smallest value of backlog plus travel time. 

We then consider the total delay (backlog plus travel time) incurred by all demands in such an equilibrium assignment. We show that if we partition the network into regions (regionalization), each of which has supply at least meeting demand, then as we use more and more regions, the total delay incurred by all the demand points in equilibrium decreases: this delay reaches the minimum possible value of the optimal travel time assignment on the original metric when there are as many regions as the number of FCs. In the other extreme, when there is only one region, we show that the delay can be a multiplicative factor of total demand away from the minimum total travel time, and that such instances can be constructed even on the line metric, extending the example above. 

Next, we turn to the question of finding a regionalization for which the equilibrium delay solution has the minimum total delay for all the demands, for any given fixed number of regions. We provide several examples in line and tree metrics suggesting that solving for the optimal such regionalization, even restricting to classes of regionalization with additional structure such as a contiguity requirement, is difficult. In particular, we show that for a fixed partitioning of the demands, the optimal assignment of FCs into regions to minimize total delay might differ from the global minimum travel time assignment. Moreover, the optimal contiguous regionalization might cost arbitrarily less than the optimal contiguous regionalization using the same assignments as a global minimum-delay assignment. Not only that, not insisting on contiguous regions can also significantly reduce the total delay when compared to the best solution with the contiguity restriction.

Given this difficulty, we turn to approximating minimum-delay solutions in the simplest case when there are no contiguity requirements on the regions. In general Euclidean metrics, we show that even a logarithmic number of regions achieves a solution with only a small multiplicative factor more total delay than the optimal total travel time in such metrics. The regions use ideas from the theory of network decompositions into small diameter regions.

We round out our analysis with some simulation studies on a synthetic two-dimensional instance with demands and supplies mimicking a US retailer, derived from public domain data on FC locations and population density. We introduce a new model trading off the imbalance in FC capacity and local demand between two extreme scenarios: In the balanced `Voronoi' scenario, we assign each population center to its closest FC and set the FC's supply to equal its assigned demand. In this case, the assignment minimizing travel time will be a feasible fulfillment with zero delay. In the other extreme, we give each FC equal capacity (equal to the total population demand divided by the number of FCs). Due to the asymmetry of FC locations compared to population density (arising from the fact that more densely populated areas are costlier to locate facilities), this gives a very unbalanced scenario. Our simulation traces the overhead of the total delay compared to the minimum possible total travel time as we trade off between these extremes. We then pick a single value of the imbalance mix to investigate a specific regionalization strategy in the literature. The result validates our model by showing a nearly 20\% reduction in total delay by using a regional decomposition hand crafted from the one shown in~\cite{sinha2026regionalize}.

\subsection{Related Work}

In online retail, the high value customers place on convenience and speed of fulfillment has driven research into new ways to optimize these metrics in fulfillment networks. For example, \cite{DWWY21} study how flexibility in the assignment of orders to fulfillment centers can improve service quality. \cite{XAG08} show how cost savings can be achieved by taking advantage of the real-time nature of order fulfillment to re-optimize assignments. Methods for initial inventory allocation in fulfillment centers using the greedy fulfillment strategy are studied by \cite{jasin2024inventory}. There are several surveys in this area, including on the rise of so-called ``Quick Commerce'' in India by \cite{ranjekar2023rise}, and fulfillment optimization in omnichannel operations by \cite{Jasin2019}. 

One recent optimization in fulfillment networks is the use of \emph{regionalization} to increase fulfillment speed. The idea of enforcing structure, thus limiting flexibility, was used by \cite{jordan1995principles} to increase the efficiency of manufacturing processes. Building on this idea, \cite{sinha2026regionalize} demonstrate the effectiveness of splitting Amazon's fulfillment network into regions. However, there has been limited research on the fundamental reasons for the observed improvement due to regionalization. 

Our characterization of equilibrium backlogs uses linear programming duality. There is a rich history of the use of duals for assignment problems, with applications to algorithms (\cite{kuhn55,Bertsekas1981,Balinski85} and generalizations to network flow~\cite{EK72}, to list just a few examples), pricing in economic markets (e.g., \cite{Shapley1971, demange86}), and more. One of the first examples is the ``Hungarian Method'' by \cite{kuhn55} to give an efficient combinatorial algorithm for the minimum cost assignment problem. \cite{demange86} use duals in a pricing algorithm to describe equilibria in multi-item auctions. They show that there is a unique price-optimal equilibrium, similar to our findings for backlogs but with a different underlying assignment problem.

\section{Preliminaries and Problem Statement}

\paragraph{Equilibrium Solutions and Delay} We denote the set of demands as $\D$ and the set of FCs as $\F$. We will always use $k$ to refer to the number of FCs, $\abs{\F}$. We suppose that $\D$ and $\F$ exist in a metric $\ell$, so that the length between a demand $i$ and FC $j$ is denoted $\ell_{ij}$. This represents the speed or travel time for service, that is, the time it takes FC $j$ to satisfy demand $i$. 

Each FC $j$ has a capacity $C_j$ of how much demand it can satisfy, and each demand $i$ has a demand of $\Dem_i$, and we will assume that $\Dem := \sum_i \Dem_i \leq \sum_j C_j$. 
One demand node $i$ can have its total demand satisfied by multiple different FCs. We will call an assignment of minimum total $\ell$-distance a \emph{minimum cost assignment}. This represents the minimum total distance incurred by all the demands in any satisfying assignment.

Turning to a greedy assignment, if we allow every demand to greedily choose an FC that minimizes $\ell_{ij}$, some FC $j$ may become overwhelmed, that is, the demand for it will exceed its capacity $C_j$. In this case, it will begin to accumulate \emph{backlog}, denoted $\beta_j$, so that the time it takes $j$ to service any demand $i$ is now $\ell_{ij} + \beta_j$. If enough backlog accumulates, demands may now wish to be served by a different FC in order to minimize service time. We are interested in equilibrium solutions to this dynamical system.

Note that the inputs to the dynamical system that we are modeling are demand and supply rates at locations, and during the arrival process, demand is routed to the FC with the lowest value of travel time plus accumulated backlog at the FC. The backlog of an FC represents the amount of time it would take the FC with supply rate per unit time equal to its capacity to ship an order for a new demand request (assuming FIFO fulfillment). If we measure these backlogs at FCs in the same time units as the travel time, then at equilibrium, we are looking for static backlog values at FCs to support an assignment that matches each demand with a minimum total delay (travel time plus backlog) FC. Thus, henceforth in this paper, we convert the demand and capacity rates to static values, and look for a backlog value in the same units as the travel times, but no longer concern ourselves with the time variations of the dynamical system.

\begin{definition}\label{def:equil}
    An \emph{equilibrium solution} is a set of backlogs $\beta_j \geq 0$ and an assignment of demand to FCs such that 1) each FC $j$ has at most $C_j$ demand assigned to it, 2) each demand $i$ is assigned only to FCs $j \in \argmin_{j \in \F} \ell_{ij} + \beta_j$, and 3) each FC $j$ with \emph{less than} $C_j$ demand assigned to it has $\beta_j = 0$. 
\end{definition}

Specifically, we are interested in equilibrium solutions in which demands are served as quickly as possible. For this, we define a notion of delay. 
\begin{definition}
    The \emph{delay} of an equilibrium solution is the sum of delays $\delta_i := \min_{j \in \F} \ell_{ij} + \beta_j$ experienced by each unit of demand. That is, $\sum_i D_i \delta_i$. 
\end{definition}

\paragraph{Regionalized Solutions: }

Due to the greedy nature of an equilibrium solution, it might require FCs to have backlogs in order to maintain equilibrium. Hence, even if the assignment used is a minimum cost assignment, the overall delay incurred could be much higher than the assignment cost. On the other hand, a minimum cost assignment may be enforced without any backlogs by sacrificing the equilibrium condition. This results in a solution with overall delay equal to the assignment cost; however, some demands may experience higher individual delay than they would upon making a local switch to a different FC. The greedy online nature of order management systems makes forcing such a matching logistically difficult so instead we consider an in-between: we split the space of demands into regions which are each restricted to the use of some subset of FCs. Then, an equilibrium solution is found within each region.

\begin{definition}
    An \emph{$r$-regionalized solution} is a partition of $\D$ and $\F$ into $r$ parts each and a pairing of the parts so that within each pair, $\sum_i \Dem_i \leq \sum_j C_j$ and the assignment within each pair is an equilibrium solution. 
\end{definition}

It may be natural to assume that the parts in the partition of $\D$ also obey some structural properties in the metric. For example, we may wish they were connected, convex, etc.

\paragraph{Problem Statements and Results.}

With these notions defined, we can formally state our main results. In \Cref{sec:characterization}, we characterize the backlogs in an equilibrium solution. We describe a linear program \eqref{eq:primal} and its dual \eqref{eq:dual}, then prove our main characterization theorem. 

\begin{restatable}{theorem}{equiloptimal}
\label{thm:equil=optimal}%
    The set of equilibrium solutions is precisely the set of optimal primal/dual solution pairs to \eqref{eq:primal} and \eqref{eq:dual}. That is, every equilibrium solution is a minimum cost assignment paired with backlogs equal to the $\beta$'s from an optimal dual solution to \eqref{eq:dual}, and vice versa.  
\end{restatable}
Based on this characterization, we then present an alternate efficient combinatorial algorithm (based on shortest paths with negative weight arcs) to compute the minimum-delay equilibrium solution for an instance. 

In \Cref{sec:regionalization}, we investigate the effects of regionalization. We first prove that $k$ regions is sufficient to get the minimum possible delay. (Recall that $k = |\cal{F}|$.) 
\begin{restatable}{theorem}{kregmincost}
\label{thm:k-regions=min-cost}%
The delay incurred by an optimal $k$-regionalized solution in an instance with $D_i = 1$ is equal to the total cost of a minimum cost assignment. 
\end{restatable}
Then we give an upper bound on the benefit that can be realized by regionalizing vs the global (1-regionalized) solution. 
\begin{restatable}{theorem}{demupperbound}
\label{thm:dem-upper-bound}%
The optimal 1-regionalized solution has delay at most a $\Dem$-factor more than the cost of the minimum cost assignment.
\end{restatable}
We show that the above lower bound can be achieved even in an instance on a line metric. 

In \Cref{sec:specific-metrics}, we provide some examples suggesting that solving for the minimum-delay regionalization exactly may be difficult. However in the line metric, we can give approximately optimal regionalizations using logarithmically many regions in the aspect ratio $\rho$ of the metric (defined as the ratio of the maximum to minimum distance between any pair of demand and facility points in the metric). 
\begin{restatable}{theorem}{constapproxeuclidean}
\label{thm:const-approx-euclidean}%
In the line metric with unit demands and capacities, there is a regionalized solution using $3 \log_2 \rho$ regions with total delay at most six times the minimum cost of any assignment. 
\end{restatable}
This result has a natural generalization to higher dimensional Euclidean metrics that we also provide.

Finally, in \Cref{sec:experiments}, we conclude with some experimental results that demonstrate an application of our model to a synthetic instance meant to emulate a nationwide US e-retailer.

\section{Characterization of Backlogs}\label{sec:characterization}

In this section we shed light on the structure of equilibrium solutions and their relationship to minimum cost assignments. This allows us to efficiently compute the minimum-delay equilibrium solution for any instance. 

The structure arises from investigating the dual to the linear program describing the problem of finding a minimum cost assignment of demand to FCs. These primal and dual linear programs can be written as follows.

\leqnomode
\noindent\begin{minipage}[t]{0.45\textwidth}
    \begin{alignat*}{3}
        \tag{$\ms{P}$}\label{eq:primal}\\
        \lpobjective{\min}{\sum_{ij} \ell_{ij} x_{ij}}\\
        \stconstraint{\sum_{j} x_{ij}}{=}{D_i}{\forall i \in \D}\\
        \lpconstraint{\sum_i x_{ij}}{\leq}{C_j}{\forall j \in \F}\\
        \lpconstraint{x}{\geq}{0}{}
    \end{alignat*}
\end{minipage}
\quad \begin{minipage}[t]{0.5\textwidth}
    \begin{alignat*}{3}
        \tag{$\ms{D}$}\label{eq:dual}\\
        \lpobjective{\max}{\sum_{i \in \cal{D}} D_i \delta_i - \sum_{j \in \cal{F}} C_j \beta_j}\\
        \stconstraint{\delta_i - \beta_j}{\leq}{\ell_{ij}}{\forall ij \in \D \times \F}\\
        \lpconstraint{\beta}{\geq}{0}{}
    \end{alignat*}
\end{minipage}
\reqnomode

\equiloptimal*
\begin{proof}
    Consider any optimal primal/dual solution pair $x$ and $\delta, \beta$. First, observe that since $\beta_j \geq 0$ for every $j$, the $\beta$'s can be interpreted as backlogs, and $x$ gives a minimum cost assignment of demands to FCs. 

    We claim that this is an equilibrium solution using backlogs $\beta_j$. For every $i \in \D$, if $i$ is assigned to FC $j$, then complementary slackness implies that $\delta_i = \ell_{ij} + \beta_j$ while for every other $j' \in \F$, the dual constraint gives $\delta_i \leq \ell_{ij'} + \beta_{j'}$. So $j$ minimizes $\ell_{ij} + \beta_j$, as desired in an equilibrium solution. Moreover, for any FC $j$ with $\sum_i x_{ij} < C_j$, complementary slackness implies that $\beta_j = 0$. Hence, $x$ and $\beta$ together satisfy the equilibrium conditions in \Cref{def:equil}.

    The other direction is similar: it is easy to verify that the backlogs for any equilibrium solution give a feasible dual. Moreover, complementary slackness is satisfied with the primal defined from this equilibrium solution, so both the primal and dual solutions must be optimal. Hence, this is a minimum cost assignment. 
\end{proof}

With this characterization theorem in hand, we can construct a new linear program that describes the problem of finding the backlogs for an equilibrium solution that minimizes the total delay. 
\begin{corollary}\label{thm:efficient-min-delay}
A minimum-delay equilibrium solution can be computed in polynomial time.
\end{corollary}

\begin{proof}
We can first find the $\ell$-cost of a minimum cost assignment by solving the assignment LP \eqref{eq:primal}. Call the cost of this solution $\OPT$. Then solve the following LP that finds the optimal dual whose delay is minimum: 
\begin{alignat*}{3}
\lpobjective{\min}{\sum_{i \in \D} D_i \delta_i}\\
\stconstraint{\delta_i - \beta_j}{\leq}{\ell_{ij}}{\forall ij \in \D \times \F}\\
\lpconstraint{\beta}{\geq}{0}{} \\
\lpconstraint{\sum_i D_i \delta_i - \sum_j C_j \beta_j}{=}{\OPT}{}
\end{alignat*}
By \Cref{thm:equil=optimal}, any optimal dual solution can be used as the backlogs for any minimum cost assignment to form an equilibrium solution. Hence, the $\beta_i$ values that arise from the solution to this linear program can be used with any minimum cost assignment to form an equilibrium solution. Optimality in the above linear program ensures that the delay of this equilibrium solution will be minimum. 
\end{proof}

In fact, it is immediately obvious that we could use a linear program to optimize any linear function of the delays and backlogs. 
\begin{corollary}
An equilibrium solution optimizing any linear function of the delays and backlogs can be computed in polynomial time. 
\end{corollary}

We now present an alternate combinatorial method to find a minimum-delay equilibrium solution with better practical computational efficiency. This method requires that each demand node be fulfilled by only one FC. However, any instance can be made to satisfy this by breaking up the demands based on how they are assigned in an integral minimum cost assignment. Formally, if $x$ is an integral minimum cost assignment, then each demand $i \in \D$ can be split into at most $k$ (recall that $k$ is the number of FCs) demands $i_j$ having $D_{i_j} = x_{ij}$ units of demand. This results in an equivalent instance whose size is polynomial in $k$ and $\abs{\D}$, so the following procedure is a strongly polynomial time way to compute the minimum-delay equilibrium. 

The procedure proceeds as follows. Fix an integral minimum cost assignment of demands to FCs, $x$, such that each demand is entirely assigned to exactly one FC. We define the weighted directed graph $\mathcal{G}_x$ whose vertex set consists of a root $r$ and the set of FCs $\F$ and demands $\D$. There is an arc of weight 0 in $\mathcal{G}_x$ from $r$ to every FC $j \in \F$, as well as arcs from $j$ to $i$ for every demand $i$ assigned to $j$ in $x$, with weight equal to $-\ell_{ij}$. For each $j \in \F$ and each $i$ \emph{not} assigned to $j$, there is an arc from $i$ to $j$ with weight $\ell_{ij}$. See \Cref{fig:digraph-char}. Note that if we think of the assignment $x$ as describing a flow from the facilities to the demands they satisfy in $x$, the graph $\mathcal{G}_x$ is simply the residual graph of this flow.
\begin{figure}[h!]
\begin{centering}
\includegraphics[scale=0.7]{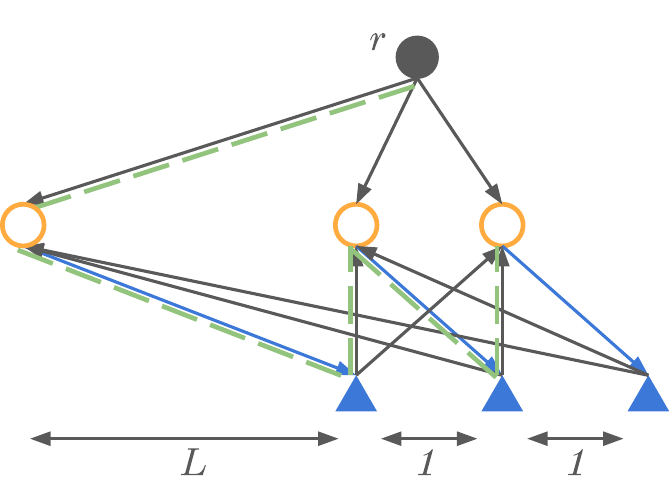} 
\caption{An example of the graph $\mathcal{G}_x$. The downward blue arcs come from the minimum cost assignment. The cost of arcs from $r$ is 0, and the cost of upward arcs is $\ell_{ij}$ while the cost of downward blue arcs is $-\ell_{ij}$. The dashed green path from $r$ shows the minimum weight path to the rightmost FC, whose backlog would be $\beta=L + 1$.}\label{fig:digraph-char}
\end{centering}
\end{figure}

\begin{theorem}
Fix a minimum cost assignment $x$. For every $j$, let $-\beta_j$ be the minimum weight of a directed path in $\mathcal{G}_x$ from $r$ to $j$, and similarly for each $i$ and $-\delta_i$. Then the resulting $\delta, \beta$ along with the assignment $x$ form a minimum-delay equilibrium solution. 
\end{theorem}

\begin{proof}
First, we note that $\mathcal{G}_x$ cannot have any negative cycles, since this would correspond to an alternating cycle of assigned and unassigned edges along which one unit of demand could be swapped in order to decrease the cost of the assignment, contradicting the optimality of $x$. Hence, minimum weight paths in $\mathcal{G}_x$ are well-defined and can be found efficiently (e.g., via the Bellman-Ford algorithm). 

We claim that the resulting $\beta$'s are backlogs in a feasible equilibrium solution. First, observe that by the definition of $\mathcal{G}_x$ and the fact that we use minimum weight paths, for each $j \in \F$ we must have that $-\beta_j$ satisfies $-\beta_j \leq -\delta_i + \ell_{ij}$ for every $i \in \D$ \emph{not} assigned to $j$. For every $i \in \D$ assigned to $j$, the arc $(j, i)$ is the \emph{only} arc entering $i$,
so these must satisfy $-\delta_i = -\beta_j - \ell_{ij}$. In either case, we have for all $i$ that 
\[
    \delta_i - \beta_j \leq \ell_{ij}.
\]
This is precisely the dual constraint from \eqref{eq:dual}. It is not hard to see that also $\beta_j \geq 0$ for all $j$, and therefore the $\delta, \beta$ pair is a feasible dual solution to \eqref{eq:dual}. Moreover, the dual constraints are satisfied with equality for every $i, j$ for which $x_{ij} > 0$, so complementary slackness is satisfied and the $\delta, \beta$ pair is also an optimal dual solution. \Cref{thm:equil=optimal} then implies that $x$ and the $\beta$'s are an equilibrium solution. 

Finally, we argue that this equilibrium solution has minimum total delay $\sum_i D_i \delta_i$. In fact, we show the stronger statement that it minimizes $\delta_i$ \emph{simultaneously} for every $i \in \D$, among all equilibrium solutions. Fix $i \in \D$, and consider a minimum weight path from $r$ to $i$ used to define $\delta_i$, labeled $r, j_1, i_1, j_2, i_2, \dots, j_t, i_t = i$. In particular, we see by tracing this path that the value of $\delta_i$ is 
\[
    \delta_i = 0 + \left(\sum_{s = 1}^{t-1} \ell_{i_sj_s} - \ell_{i_sj_{s+1}}\right) + \ell_{ij_t}.
\]
But note that this path corresponds to a set of dual constraints that \emph{every} equilibrium solution must satisfy. Moreover, by \Cref{thm:equil=optimal}, every equilibrium solution must satisfy the constraints corresponding to arcs $(j_s, i_s)$ with equality by complementary slackness with $x$. Hence, for an arbitrary equilibrium solution with delays and backlogs $\delta', \beta'$ respectively, these must satisfy 
\[
    \ell_{i_sj_s} = \delta'_{i_s} - \beta'_{j_s} \leq \beta'_{j_{s+1}} + \ell_{i_sj_{s+1}} - \beta'_{j_s}
\]
for each $1 \leq s \leq t-1$. Rearranging and summing these constraints, we get
\begin{align*}
    \delta'_i &= \beta'_{j_t} + \ell_{ij_t}\\
    &\geq \left(\sum_{s = 1}^{t-1} \beta'_{j_{s}} + \ell_{i_sj_{s}} - \beta'_{j_{s + 1}} - \ell_{i_sj_{s+1}} \right) + \beta'_{j_t} + \ell_{ij_t}\\
    &= \beta'_{j_1} + \left(\sum_{s = 1}^{t-1} \ell_{i_sj_s} - \ell_{i_sj_{s+1}}\right) + \ell_{ij_t}\\
    &\geq \delta_i
\end{align*}
where the second equality is by telescoping, and the final inequality follows from $\beta'_{j_1} \geq 0$, and the definition of $\delta_i$. 

Since this holds for all $i \in \D$, and for any arbitrary equilibrium solution $\delta', \beta'$, it follows that $\delta, \beta$ is the equilibrium solution minimizing $\sum_i D_i \delta_i$. 

\end{proof}

An immediate consequence of the final part of the previous proof is that the set of delays for the minimum-delay equilibrium solution is unique, even though the assignment used might not be unique.
\begin{corollary}
The set of delays for a minimum-delay equilibrium solution is unique. 
\end{corollary}

The delays and backlogs must obey complimentary slackness with every minimum cost assignment, and hence by fixing any minimum cost assignment, we can see that the backlogs are entirely determined by the delays in an equilibrium solution. Therefore it immediately follows that in the minimum-delay solution, the backlogs are also unique.

\begin{corollary}
    The set of backlogs for a minimum-delay equilibrium solution is unique. 
\end{corollary}

\section{Regionalization}\label{sec:regionalization}

In this section, we investigate the extent to which regionalization can be used to decrease the delay incurred in an instance. First, we make an easy observation:  

\begin{lemma}\label{lem:min-cost-LB}
    The minimum cost assignment is a lower bound on the delay of \emph{any} regionalized solution.
\end{lemma}

This simply follows from the fact that a regionalized solution still must satisfy all of the demand, and even with no backlogs, the total delay will be at least the minimum cost of any assignment. 

Our first result is the observation that with sufficiently many regions in an instance with unit demands $D_i = 1$, the delay can be decreased to the minimum possible value, corresponding to all FCs having zero backlog. Recall that $k$ refers to the number of FCs.

\kregmincost*

\begin{proof}
    This can be seen by partitioning $\F$ into $k$ parts, each containing exactly one FC. Then partition $\D$ into $k$ parts based on how the demands are assigned to the $k$ FCs in a minimum cost assignment. Pair up the parts according to this minimum cost assignment.

    Observe now that no backlogs are needed since each part is paired up with a single FC, so an equilibrium solution can only assign the demand in that part to that FC. Hence, $\delta_i = \ell_{ij}$ for each $i$ assigned to FC $j$, so the total delay is the cost of the assignment. 
\end{proof}

Next, we show that without regionalization, the total delay cannot be much higher than the cost of the minimum cost assignment. Recall that $\Dem$ denotes the total demand $\sum_i \Dem_i$. 

\demupperbound*
\begin{proof}
    To prove this, we investigate the structure of the duals for the assignment LP. Consider an optimal 1-regionalized solution and recall from \Cref{thm:equil=optimal} that this corresponds to an optimal primal/dual pair $x$ and $\delta, \beta$. Consider the graph $T$ of tight edges in the assignment LP. That is, those $i, j$ pairs for which 
    \[
        \delta_i = \ell_{ij} + \beta_j
    \]
    We first claim that every connected component of $T$ has some FC $j$ with $\beta_j = 0$. Otherwise, we could uniformly reduce the $\delta$'s and $\beta$'s for all demands and FCs in that component by some small amount without violating any dual constraints or changing the objective. However, this would reduce the \emph{delay} $\sum_i D_i \delta_i$ of the solution, contradicting the original assumption that the 1-regionalized solution is optimal. 

    Fix an (integral) minimum cost assignment $x$ whose cost we denote by $\cost(x)$. Within a connected component in $T$, for any demand $i$ in the component, there is a path $e_1, e_2, \dots$ of tight edges from $i$ to an FC $j$ with  $\beta_j = 0$. And since all demand is served in the assignment, we can assume that each odd edge $e_1, e_3, \dots$ is an assignment edge used by $x$. Since all of these edges are tight, we have 
    \[
        \delta_i = \sum_{s \text{ odd}} \ell_{e_s} - \sum_{s \text{ even}} \ell_{e_s} \leq \sum_{s \text{ odd}} \ell_{e_s} \leq \cost(x)
    \]
    Hence, $\sum_i D_i \delta_i \leq \sum_i D_i \cost(x) = \Dem \cdot \cost(x)$. 
\end{proof}

In fact, we can show that the above theorem is tight up to constant factors by exhibiting an example instance on a line metric. 

\begin{theorem}\label{thm:line-LB}
    There is an example on the line metric in which the delay of the optimal 1-regionalized solution is $\Omega(\Dem)$ more than the delay of the optimal $k$-regionalized solution. 
\end{theorem}
\begin{proof}
Consider a line divided into $k$ segments. The leftmost $k-1$ segments are of unit length, and the rightmost segment has length $L + 1$. The leftmost segment contains $D'$ units of demand at its left endpoint and an FC of capacity $D'$ at its right endpoint. Every other segment contains a single unit of demand at its left endpoint and an FC with capacity 1 on its right endpoint. Label the demands and FCs $i_1, i_2, \dots$ and $j_1, j_2, \dots$, respectively, from left to right. See \Cref{fig:line-LB}.

\begin{figure}[h!]
\begin{centering}
\includegraphics[scale=0.7]{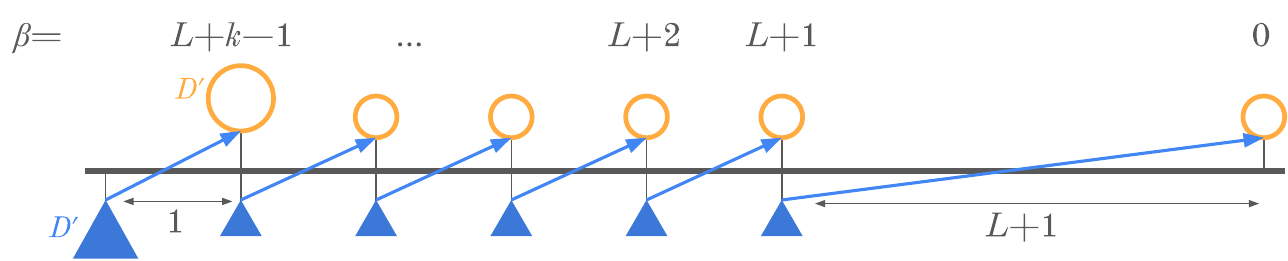} 
\caption{The instance on the line demonstrating a high-delay 1-regionalized solution. Triangles are demands, and circles are FCs. The leftmost demand node and FC have demand and capacity $D'$, respectively. Backlogs are shown above each FC.}\label{fig:line-LB}
\end{centering}
\end{figure}

In this instance, the optimal assignment just assigns $i_s$ to $j_s$ for every $s = 1, \dots, k$. In the 1-regionalized solution, this requires backlogs $\beta_{s} := L + (k-s)$ for $s \leq k-1$, and $\beta_{k} = 0$. Then the delay of the 1-regionalized solution is 
\begin{align*}
    &D'(\beta_1 + 1) + \sum_{s = 2}^{k-1} (\beta_s + 1) + \beta_k + (L + 1) \\
    =\ & D'(L +  k) + (k-2)(L + 1) + \frac{(k-1)(k-2)}{2} + L + 1\\
    =\ & D'(L +  k) + (k-1)(L + 1) + \frac{(k-1)(k-2)}{2}
\end{align*}
whereas the delay of the $k$-regionalized assignment is just $D' + k - 1 + L$. With $L >> D' >> k$, we find that the $k$-regionalized solution has delay $O(L)$, whereas the 1-regionalized solution has delay $\Omega(D'L)$. Moreover, $\Dem = \Theta(D')$, and hence the 1-regionalized solution has a factor $\Omega(\Dem)$ more delay. 
\end{proof}

\section{Regionalization for Specific Metrics}\label{sec:specific-metrics}

\subsection{Difficult Examples}

Solving for the optimal regionalization for any number of regions greater than 1 and less than $k$ appears difficult. 
It involves finding the optimal partition of the demands and FCs so that the overall delay among all regions is minimized. 

Therefore, we might naturally attempt to impose additional structure on the regions in order to make the search space of partitions more tractable. In particular, it is natural to require that the demands in a region satisfy a type of ``contiguity'' condition. Under a sufficiently restrictive condition, we might hope that finding the best way to assign FCs to regions to serve these demands could be efficiently computed. The simplest possible assignment would be the one that agrees with the global minimum cost assignment. Unfortunately, in this subsection, we demonstrate examples on simple metrics in which this assignment yields a solution that is non-optimal for total delay. 

First, we demonstrate this for line metrics. 
\begin{claim}
    There is an instance in the line metric such that, for a fixed partitioning of the demands, the optimal assignment of FCs to regions to minimize total delay is not given by the global assignment that minimizes the total cost.
\end{claim}

\begin{proof}
See \Cref{fig:line-non-contig}. All demands and FCs have unit demand and capacity. The demands have been fixed into a partition placing the left demand node in its own region, and the right two demand nodes in a region. 
\begin{figure}[h!]
\begin{centering}
\includegraphics[scale=0.6]{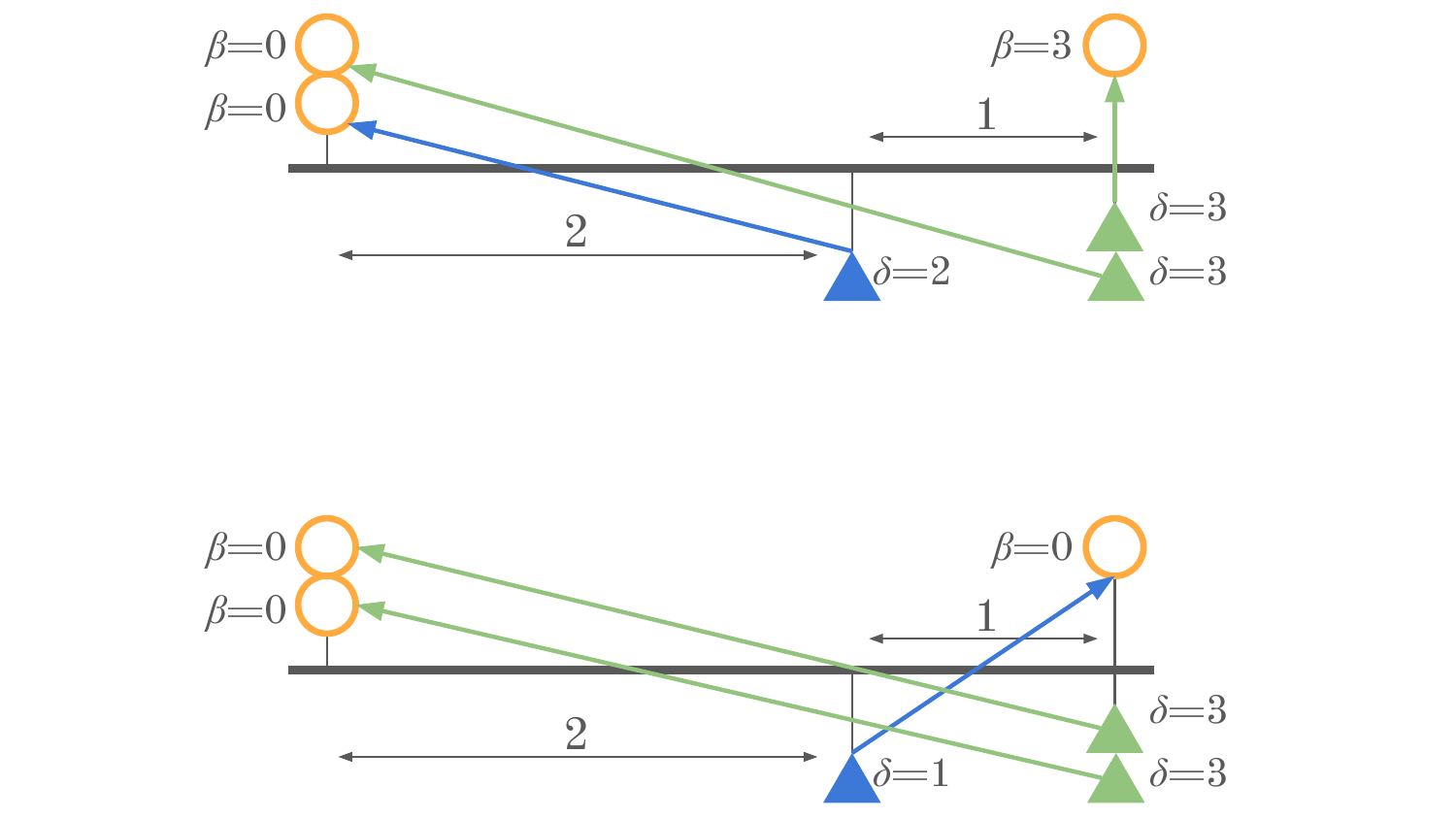} 
\caption{In both examples, the partitioning of the demands (triangles) into regions is fixed. The top example shows the regionalization resulting from assigning FCs (circles) based on the global minimum cost assignment. This results in a total delay of 8, while the bottom example shows the optimal regionalization of the FCs with a total delay of 7.}\label{fig:line-non-contig}
\end{centering}
\end{figure}
\end{proof}

Even though the above example shows that the global assignment is not useful for a fixed partition of the demands, note that if we are required to find the optimal delay solution with two regions, we could still repurpose the global assignment: we would split one of the two demands in the right in its own region and leave the other two in the second region, and use the global assignment. This alternate split of the demands results in a zero backlog solution that still uses the global assignment. An intriguing open question is whether for any value of the number of regions, we can always repurpose some optimal global assignment and partition the assigned pairs so that the resulting regionalization has minimum total delay in the line metric.

Next, we discuss tree metrics. Specifically, we give some examples that suggest that efficient computation of optimal regionalized solutions on general tree metrics may be hard. For tree metrics, we say that a regionalization is \emph{contiguous} if the trees induced by the set of demands in each region are edge-disjoint. 

First, we observe that an optimal regionalized solution might use a different assignment than the global minimum cost assignment. 

\begin{claim}
    There is an example in tree metrics in which the optimal contiguous 2-regionalization uses a different assignment than the unique optimal global assignment. 
\end{claim}
\begin{proof}
See \Cref{fig:tree2-2-parts}. 
\begin{figure}[h!]
\begin{centering}
\includegraphics[scale=0.7]{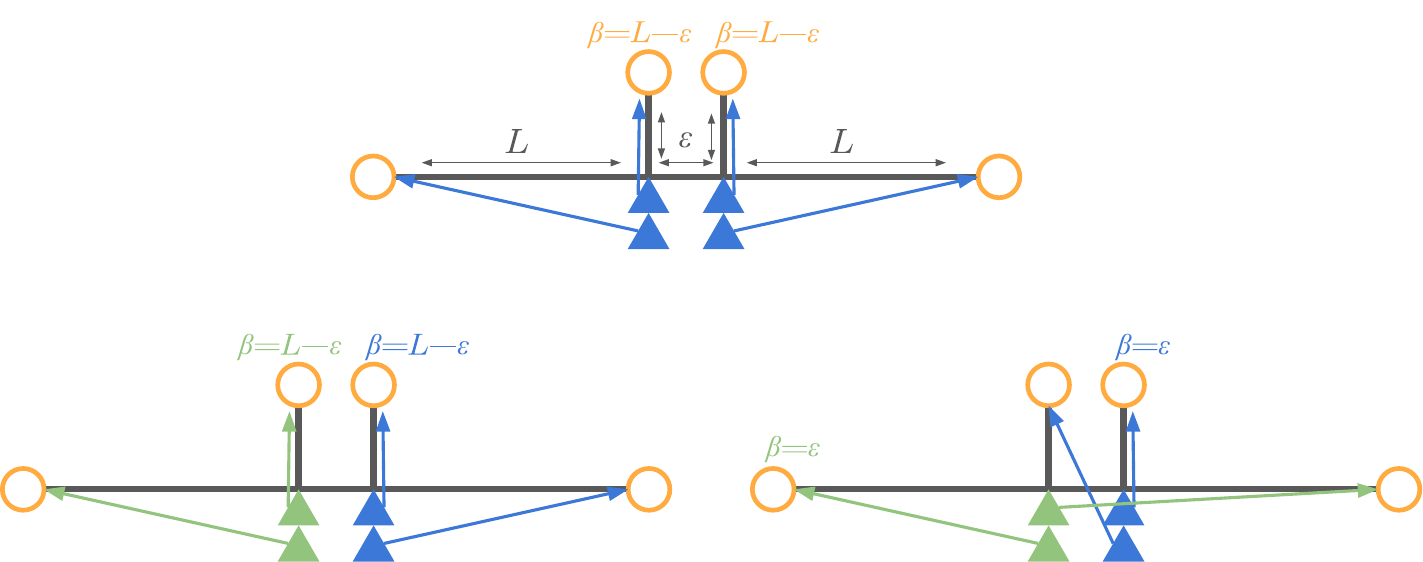} 
\caption{The optimal global assignment in this instance is shown on top. However, the best contiguous 2-regionalization using the same assignment (shown on the left) has delay $4L$, whereas there is a contiguous 2-regionalization using a different assignment (shown on the right) with delay only $2L+6\eps$.}\label{fig:tree2-2-parts}
\end{centering}
\end{figure}
\end{proof}

Finally, generalizing the example above, we observe that restricting to the use of the global minimum cost assignment can result in a regionalized solution that is worse than the optimal regionalized solution by a factor $\sqrt{k}$. 

\begin{claim}
    There is an example in tree metrics in which the optimal contiguous $\Theta(r)$-region assignment with connected subtrees costs a factor $r$ less than the optimal contiguous $r$-region assignment that uses the same assignment as the global assignment (where $r = \sqrt{k}$). 
\end{claim}
\begin{proof}
It is straightforward to generalize the example from \Cref{fig:tree2-2-parts} to $r$ regions. See \Cref{fig:trees-r-parts} for the instance with 3 regions. In the global assignment, the FCs at distance $\eps$ from their assigned demands must have backlog $\beta = L-\eps$, and hence the contiguous 3-regionalization using the global assignment has delay $9L$, whereas the other assignment shown for the contiguous regionalization has delay $3L + 24\eps$.
\begin{figure}[h!]
\begin{centering}
\includegraphics[scale=0.7]{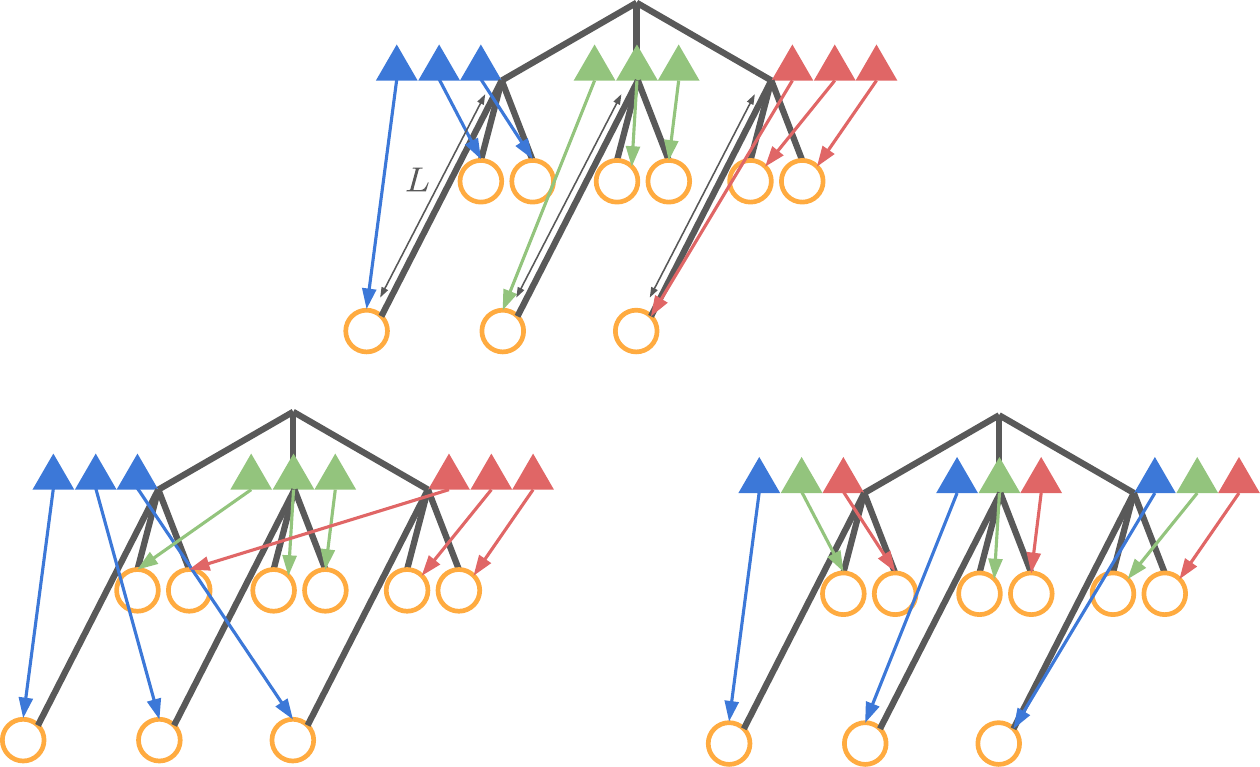} 
\caption{In this instance, the global solution is shown on top, with the natural contiguous regionalization. All distances are $\eps$ except the ones labeled as $L$. The top solution has total delay $9L$. Below on the left, an alternate assignment using the same contiguous regionalization is shown with delay $3L + 24\eps$. On the right, an alternate \emph{non-contiguous} regionalization is shown that still uses the global assignment, but has delay $3L + 6\eps$.}\label{fig:trees-r-parts}
\end{centering}
\end{figure}\end{proof}

It is worth noting that \Cref{fig:trees-r-parts} also demonstrates that there are instances such that in order to achieve low delay with the global assignment, it is necessary to use a \emph{noncontiguous} regionalization. In this instance, we get delay just $3L + 6\eps$ for the noncontiguous regionalization and global assignment shown. 

\subsection{Approximation for Euclidean Metrics}\label{sec:euclidean}
Finding optimal regionalizations appears to be difficult, even for simple metrics, and enforcing that we use the assignment from the global minimum cost assignment sacrifices significantly on the delay of the solution when we restrict to contiguous regionalizations. Therefore, we turn to regionalizations that \emph{approximate} the delay of the optimal solution, without insisting on any additional restrictions on the regions (such as contiguity). 

We first focus on instances where the metric comes from a Euclidean line. In the previous section, we described such an instance in which the optimal $k$-regionalized solution has $\Theta(\Dem)$-factor less delay than the optimal 1-regionalized solution. While this might seem to suggest that high levels of regionalization are necessary in order to achieve low delay, here we show that this is not necessarily the case. Specifically, we show that with just logarithmically (in the aspect ratio $\rho$ of the line metric) many regions, one can achieve a delay that is just a constant factor greater than that of the optimal $k$-regionalized solution. For these results, we assume that the demands and capacities are all unit $D_i = C_j = 1$.

\constapproxeuclidean*

The idea is to construct regions as well-separated clusters. By keeping the clusters that make up a single region far away from each other, we prevent the accumulation of backlogs caused by a long chain of FCs and demands that is present in the example from \Cref{thm:line-LB}. To make this more concrete, we split the demands/FCs into buckets based on the distance scale at which they are assigned in the minimum cost assignment. Then for each distance scale, we create three regions, each composed of a collection of segments of the line whose length is equal to the distance scale, and each separated by a length of twice the distance scale. The formal proof is deferred to Appendix~\ref{app:euc-proofs}.

Since the cost of the minimum cost assignment is a lower bound on the delay of the minimum-delay regionalized assignment by \Cref{thm:k-regions=min-cost}, this implies a constant approximation to the best regionalized assignment.

\begin{corollary}
    There is a factor-6 approximation for the minimum-delay regionalized solution in the line metric using $3\log_2 \rho$ regions.
\end{corollary}

We can get a similar guarantee for the Euclidean plane, or even the $q$-dimensional Euclidean space as follows:
\begin{corollary}\label{cor:q-dim-euc}
    $(2 + \sqrt{q})^q \log \rho$ regions are sufficient to get a $(4\sqrt{q} + 2)$-approximation for the optimal delay solution on instances in $q$-dimensional Euclidean space.
\end{corollary}
The proof can be found in Appendix~\ref{app:euc-proofs}.

\section{Experimental Results}\label{sec:experiments}

To investigate our findings experimentally, we created an instance mimicking a US retailer. We use US census data on ZIP code populations to create a set of demand nodes corresponding to ZIP3 cells, with each node located at the centroid of the cell and having demand proportional to the population of the cell. We used publicly available data on Amazon Fulfillment Center locations as our set of FCs. This was done using a large language model with a prompt requesting a list of the top 80 largest Amazon FCs in the contiguous United States, along with their latitude/longitude coordinates. 

\subsection{FC Capacities}

In order to determine reasonable FC capacities, we combined two heuristically derived settings. In one, all FCs are assumed to have equal capacity such that there is enough total capacity in the instance to meet all demand. In the other setting, FCs are assumed to have enough capacity to meet the demand of all those demand nodes for whom that FC is the closest. We call these ``Voronoi'' capacities, since this is equivalent to giving an FC capacity sufficient to meet all of the demand in its Voronoi cell. In this case, we gave the FCs with the least capacity some extra capacity in order to make the total capacities in both settings equal. 

We then considered instances created by taking convex combinations of these two capacity vectors. If we denote by $C^{(e)}$ the equal capacity vector and $C^{(v)}$ the Voronoi capacity vector, we study instances with capacity vectors of the form 
\[
    \alpha C^{(v)} + (1 - \alpha) C^{(e)}
\]

An initial observation is that, with Voronoi capacities, the minimum cost assignment will be precisely the assignment described by the Voronoi cells, and such an assignment is an equilibrium solution with all zero backlogs. On the other hand, the equal capacity instance may need backlogs in an equilibrium assignment, and this was the case in our instance. Figures~\ref{fig:v-equil} and ~\ref{fig:e-equil} depict the minimum-delay equilibrium solutions for the cases with Voronoi and equal capacities, respectively.

\begin{figure}[h!]
\begin{centering}
\includegraphics[scale=0.35]{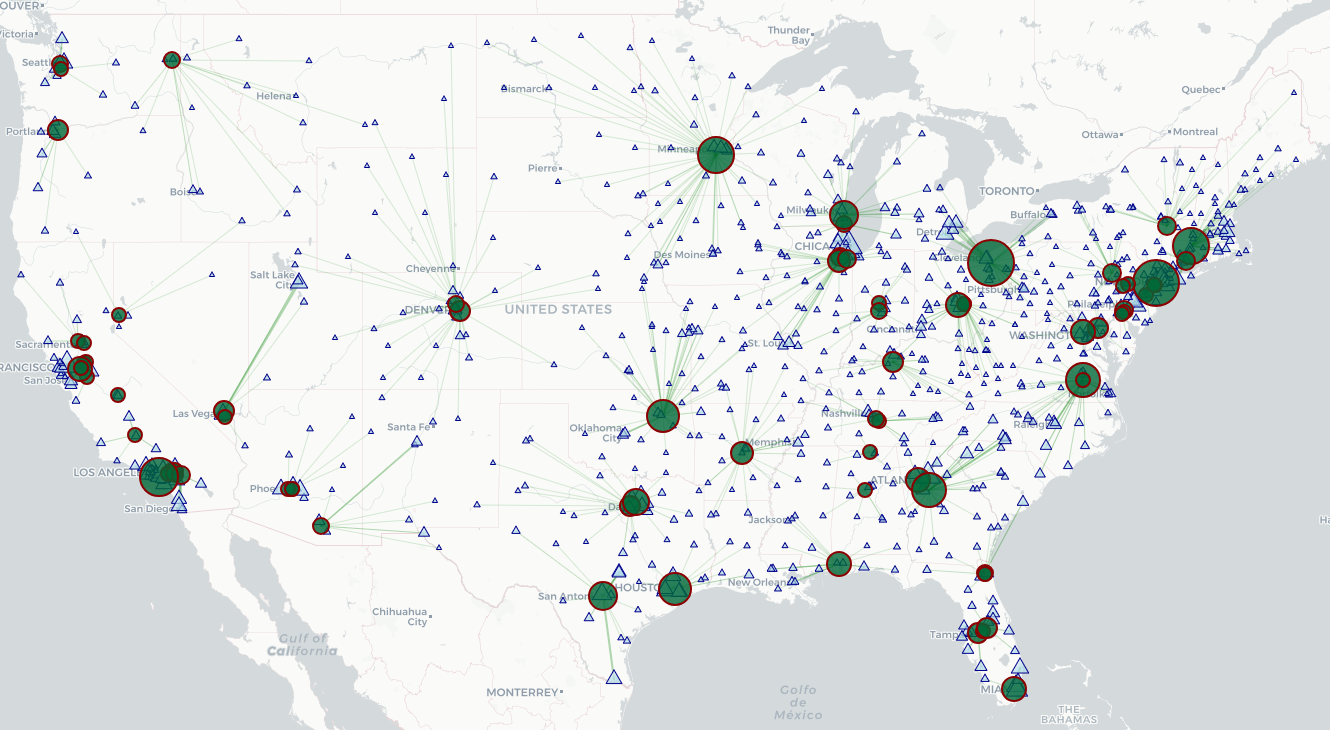} 
\caption{A minimum-delay equilibrium assignment for the Voronoi capacity instance. Triangles represent demands, and circles FCs, with size determined by demand rate/capacity. All FCs have zero backlog.}\label{fig:v-equil}
\end{centering}
\end{figure}

\begin{figure}[h!]
\begin{centering}
\includegraphics[scale=0.35]{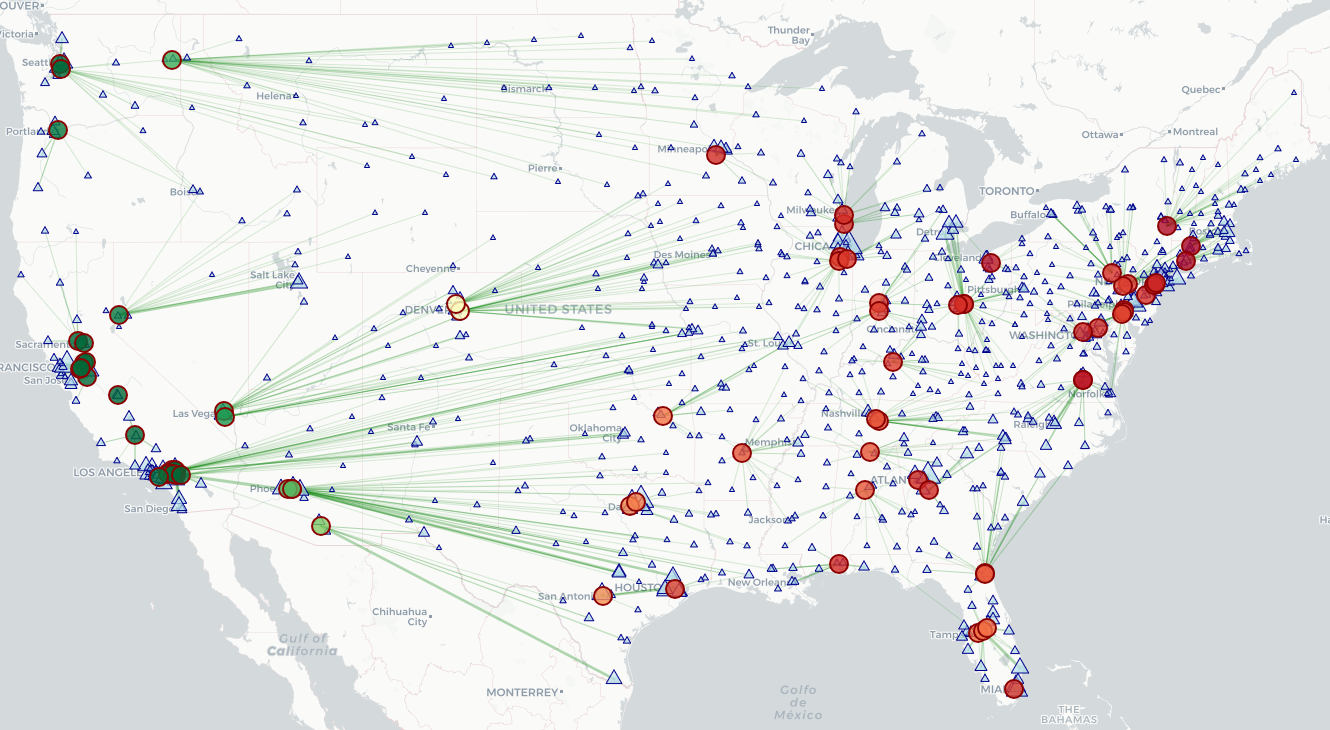} 
\caption{A minimum-delay equilibrium assignment for the equal capacity instance. Triangles represent demands, and circles FCs. Green FCs indicate low backlog, and red indicates high backlog.}\label{fig:e-equil}
\end{centering}
\end{figure}

By taking convex combinations of these capacity vectors and solving for the minimum-delay equilibrium assignment in each instance, we observed a relatively smooth trend of the overall delay decreasing as $\alpha$ increases---see Figure~\ref{fig:delay-vs-alpha}. 

\begin{figure}[h!]
\begin{centering}
\includegraphics[scale=0.5]{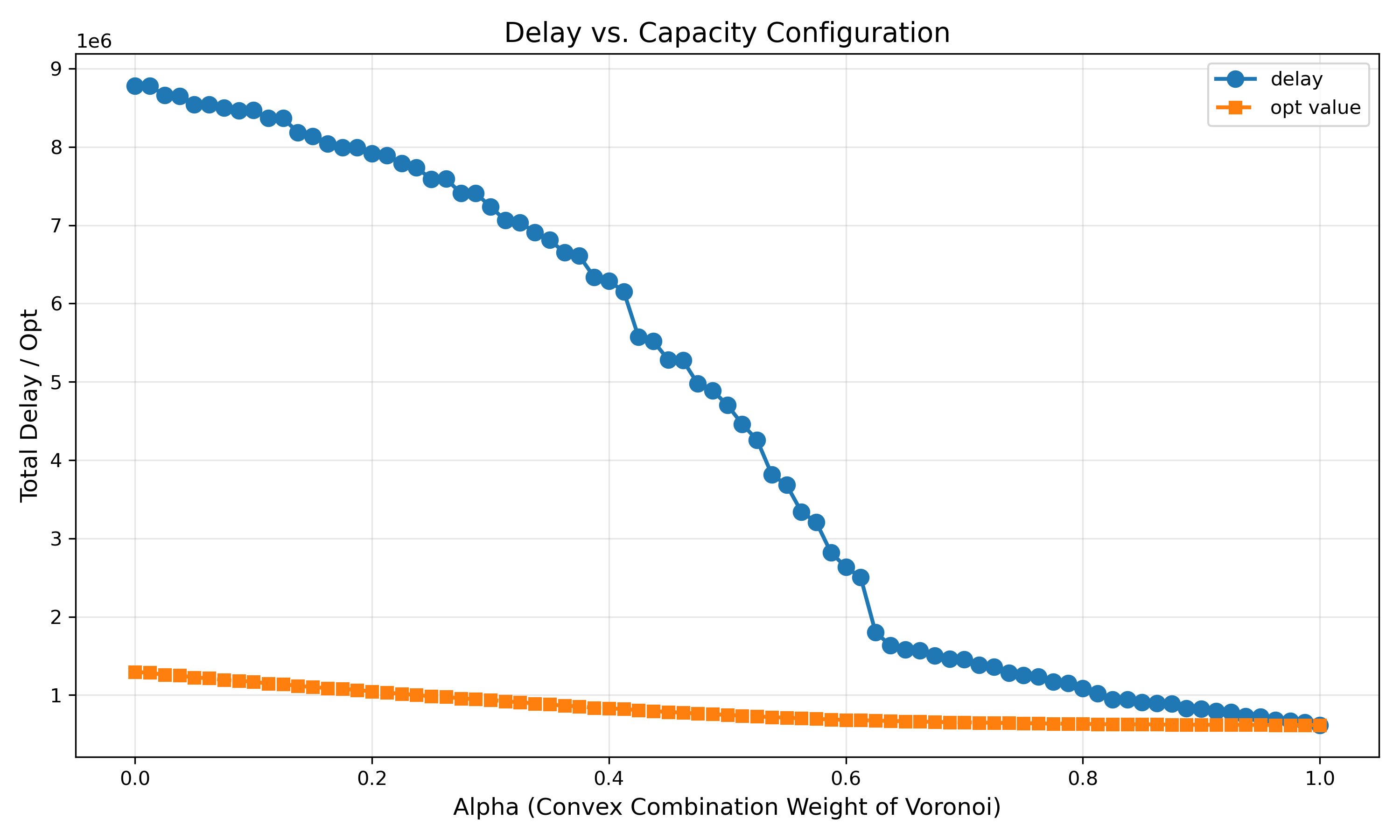} 
\caption{Delay versus $\alpha$. The orange line shows the cost of the minimum cost assignment for each instance. }\label{fig:delay-vs-alpha}
\end{centering}
\end{figure}

\subsection{Regionalization}

Next, we investigate how the overall delay is affected by simple regionalization. 
We used the capacity vector corresponding to $\alpha = 0.6$ in this example, and created a regionalization by hand that mimics an agglomeration of the regions used by Amazon from \cite{sinha2026regionalize}. The global equilibrium solution had an overall delay of 2,634,363 (measured in kilometers) and is depicted in Figure~\ref{fig:global-equil}, whereas the regionalized solution depicted in Figure~\ref{fig:regionalized-equil} had a lower delay of 2,107,233, an improvement of approximately 20\%. 

\begin{figure}[h!]
\begin{centering}
\includegraphics[scale=0.35]{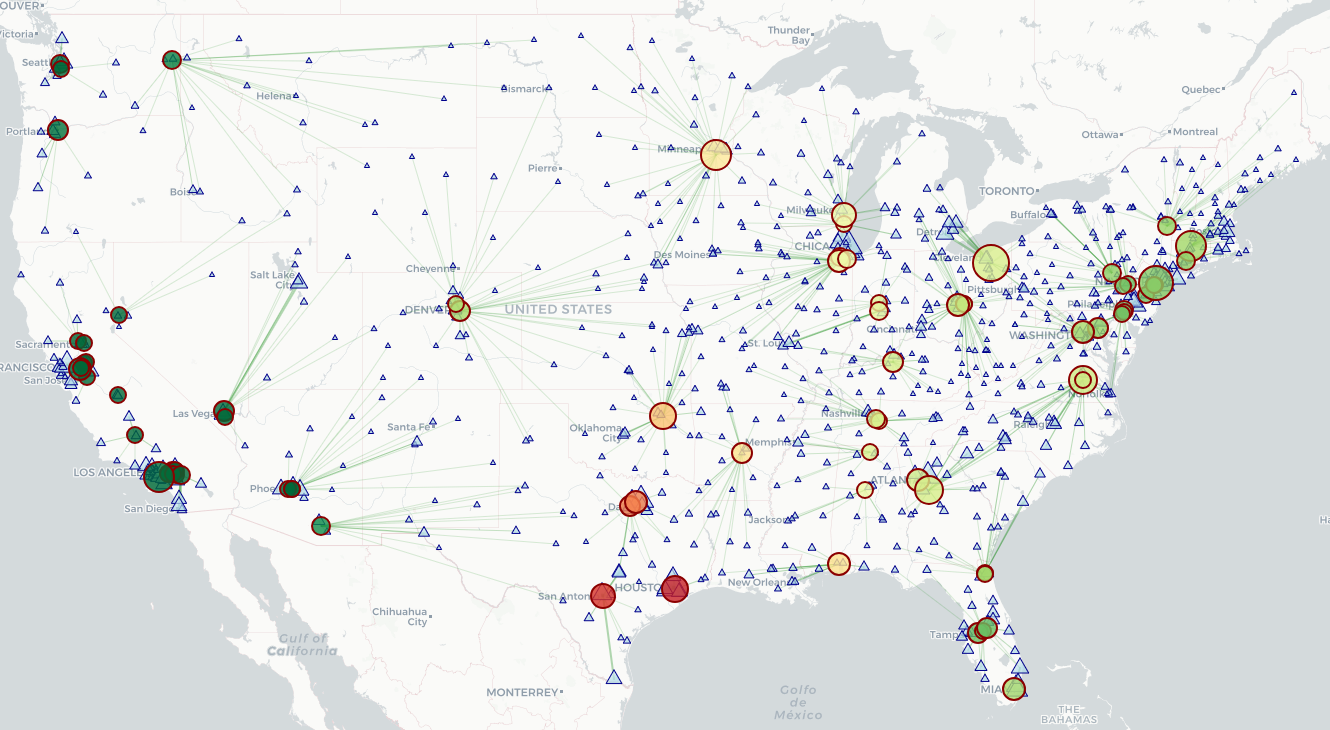} 
\caption{The global equilibrium solution with delay 2,634,363. Triangles represent demands, and circles FCs, with size determined by demand rate/capacity. FC capacities are chosen corresponding to $\alpha=0.6$. Green FCs indicate low backlog, and red indicates high backlog.}\label{fig:global-equil}
\end{centering}
\end{figure}

\begin{figure}[h!]
\begin{centering}
\includegraphics[scale=0.35]{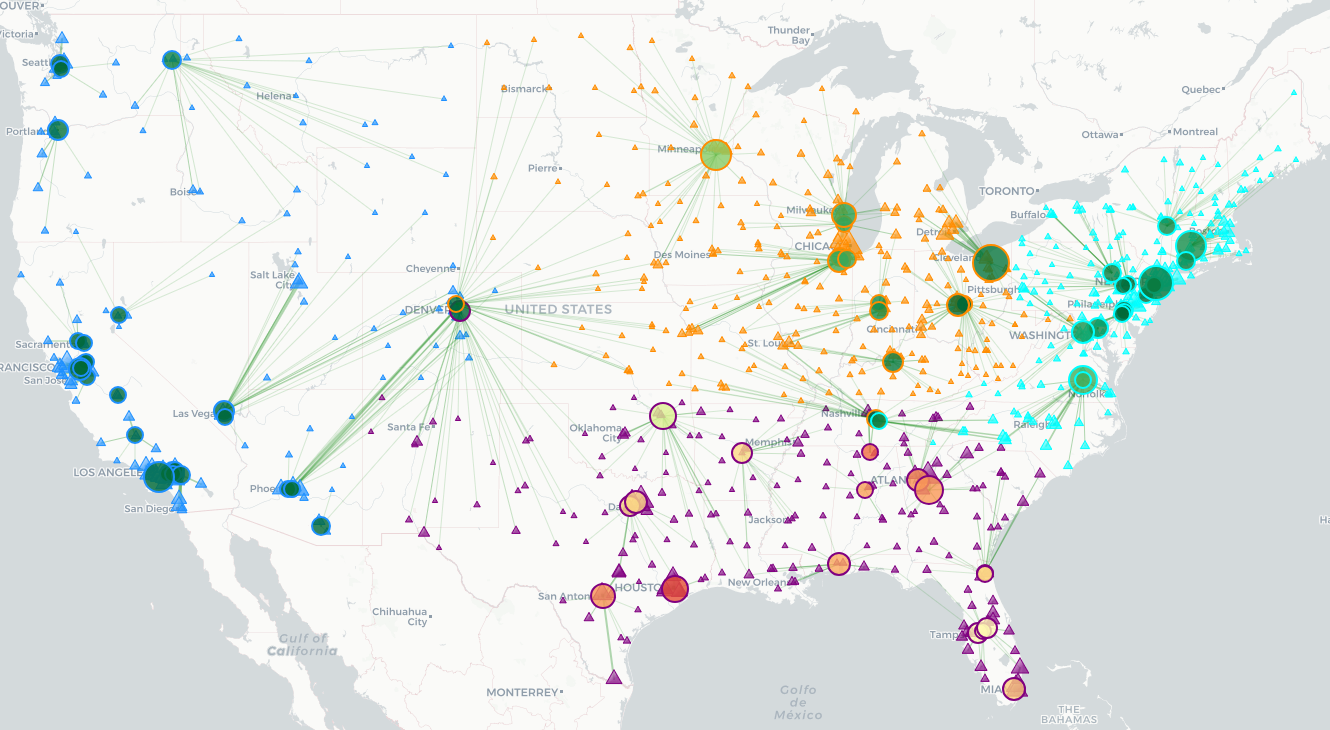} 
\caption{The Amazon-like regional solution with delay 2,107,233. Regions are represented by colors of the triangles (demands) and border colors of the circles (FCs).}\label{fig:regionalized-equil}
\end{centering}
\end{figure}

\section{Conclusion}
We have shown a simple model of equilibrium backlogs and delays in a greedy fulfillment situation modeling a commonly used e-retailing policy. We used this model to show how the total fulfillment delay can be substantially reduced by regionalizing the fulfillment network. We then show how network decomposition based methods can be used to find a small logarithmic number of regions in Euclidean metrics with total delay of the same order as the minimum travel time of any solution. Finally, we showed empirically that the model predicts a significant reduction in total fulfillment delay on a nationwide US e-retailer network, in line with the results reported from the regionalization effort. 

Many interesting questions remain.
Is the problem of optimal regionalization to minimize total delay for a given specified number of regions NP-hard in general metrics? On the other side, are there simple solutions to this problem on more structured metrics such as one or two dimensional Euclidean distances? Are the latter problems more tractable if the set of allowed regionalizations are required to be contiguous in the sense that we have used in this paper, or under other regularity conditions (such as the demand regions forming a convex set)? More specifically, we have not even been able to resolve the question of whether in a line metric, for any number of regions, there is an optimal regionalized solution for minimum total delay
that obeys the global assignment of a minimum-cost solution. Such an assignment is given, without loss of generality, by assigning the demands in the left-to-right order to the FCs in the left-to-right order filling the FCs up to their capacity.
Our work used a static model of backlogs to argue equilibrium properties. However, the dynamical system is exposed to variation in both the demand rates (due to seasonality) and supply rates (due to weather and other network maintenance events). Are the regional solutions arising from solving the static backlog problem robust to such variations? Are there any unintended effects in using such solutions for the total delay of the dynamical systems?

From a practical perspective, there are several additional dimensions to consider, some of which are noted in \cite{sinha2026regionalize}. For example, this paper models a single sku, but in reality there are a large number of skus with different demand distributions. The linear transportation cost in this paper is also an approximation; in reality, for example, the requirement to use integer numbers of vehicles for transportation brings additional computational complexity. A third aspect to consider is the ability to increase capacity at a cost (which may be non-linear); for example, by temporarily activating some inefficient process to get a small amount of additional capacity. Each of these aspects offers additional research questions.

\bibliographystyle{alpha}
\bibliography{references}

\begin{appendices}
\appendix

\section{Omitted Proofs from \Cref{sec:euclidean}}
\label{app:euc-proofs}

\begin{proof}[Proof of \Cref{thm:const-approx-euclidean}]
Let $\rho$ be the aspect ratio of the metric $\ell$ (the maximum distance divided by the minimum distance). Assume WLOG that the minimum distance is 1 by scaling, so $\rho$ is the max distance. 

Now consider a minimum cost assignment. Each demand $i$ is matched to some FC at some distance. We will bucket the demands into $D^1, D^2, \dots, D^{\log \rho}$, where a demand $i$ belongs to bucket $D^d$ if it is matched to an FC at a distance in $[2^{d-1}, 2^{d})$. 

Fix a bucket $D^d$, and split the line into contiguous segments of length $2^{d + 1}$, call them $L_1, L_2, \dots$. We define three regions $R_d^{(1)}, R_d^{(2)}, R_d^{(3)}$, where $R_d^{(\kappa)}$ is all those demands from $D^d$ that fall into a segment $L_s$ for $s \equiv \kappa \pmod{3}$, along with the FCs they are assigned to in the minimum cost assignment.

We claim that for any $\kappa$, a minimum-delay assignment of region $R_d^{(\kappa)}$ will have a demand in each segment whose assigned FC has $\beta = 0$. Suppose not and consider the segment with no $\beta = 0$ FC whose minimum $\beta$ FC is largest among all such segments. Since this is a minimum-delay assignment, there must be some tight constraint between an FC $j$ assigned to a demand in this segment and a demand $i$ in another segment, otherwise we could uniformly shift the dual values for all FC/demand pairs in this segment downward without violating any constraint, contradicting the fact that this is a minimum-delay assignment. 

But now the tight constraint implies that 
\begin{align}
    \beta_j &= \delta_i - \ell_{ij}\nonumber \\
    & < \delta_i - (2 \cdot 2^{d + 1} - 2^d)\nonumber \\
    &= \delta_i - 3 \cdot 2^d. \label{eq:beta-bound}
\end{align}
The inequality follows because these two segments are separated by a distance of at least the width of two segments and the distance scale is $[2^{d-1}, 2^d)$. Moreover, since we started with the segment maximizing the minimum $\beta$ and $i$ lies in another segment, then $i$ must be within distance $2^{d + 1} + 2^d = 3 \cdot 2^d$ of some FC $j'$ with $\beta_{j'} \leq \beta_j$ (note that $i$ may not be assigned to FC $j'$, it just lies within the same segment), and hence $\delta_i \leq \ell_{ij'} + \beta_{j'} \leq 3 \cdot 2^d + \beta_j$. 
Combining this with \cref{eq:beta-bound}, we get $\beta_j < \beta_j$, a contradiction. This finishes the proof of the intermediate claim. 

Finally, since every segment in region $R_d^{(\kappa)}$ has a demand whose assigned FC $j^*$ has $\beta_{j^*} = 0$, we get for every demand $i$ and corresponding assigned FC $j$ in that segment 
\[
    \delta_i \leq \ell_{ij^*} + \beta_{j^*} \leq \ell_{ii^*}+\ell_{i^*j^*} \leq 2^{d + 1} + 2^d= 6 \cdot 2^{d-1} \leq 6 \cdot \ell_{ij}
\]
since the segments in region $R_d^{(\kappa)}$ are of length $2^{d+1}$. Hence, the total delay is $\sum_i \delta_i \leq 6 \cdot \sum_i \ell_{ij}$ where the right hand side is the cost of a minimum cost assignment, which is a lower bound on the optimal delay by \Cref{thm:k-regions=min-cost}. So we get a 6-approximation to the optimal delay using $3 \log \rho$ regions. 
\end{proof}

\begin{proof}[Proof of \Cref{cor:q-dim-euc}]
    The proof follows as above: we again consider distance scale buckets $D^1, \dots, D^{\log \rho}$ with $D^d$ at scale $[2^{d-1}, 2^d)$. Denote $B := 2 + \ceil{\sqrt{q}}$. We now make $B^q$ regions for each distance scale by splitting the space into hypercube ``segments'' $L_{s_1, \dots, s_q}$ of side length $2^{d+1}$ and grouping them into regions based on the index modulo $B$ for each coordinate. More precisely, for each $q$-length tuple of the form $(\kappa_1, \dots, \kappa_q)$, we create a region $R_d^{(\kappa_1, \dots, \kappa_q)}$ containing all segments $L_{s_1, \dots, s_q}$ with $s_1 \equiv \kappa_1 \pmod{B}, \dots s_q \equiv \kappa_q \pmod{B}$.  

     Again, assume for contradiction that some minimum-delay assignment of a region has a segment with no $\beta=0$ FC, consider the one whose minimum $\beta$ FC is largest. Then there must be a tight constraint between FC $j$ assigned to a demand in this segment, and demand $i$ from another segment. But then observe that any two segments in this region are separated by a distance of at least $(B - 1) \cdot 2^{d + 1} \geq (\sqrt{q} + 1) \cdot 2^{d+1}$, while $j$ can lie at most the distance scale $2^d$ outside of its segment. Hence,  
    \begin{align*}
        \beta_j &=\delta_i - \ell_{ij}\\
        &< \delta_i - \bigl[(\sqrt{q} + 1) \cdot 2^{d+1} - 2^d\bigr]\\
        &= \delta_i - (2\sqrt{q} + 1) \cdot 2^d 
    \end{align*}
    Each segment has a diameter of $\sqrt{q} \cdot 2^{d + 1}$, so as before, since we chose the segment maximizing the minimum $\beta$, then $i$ must be within distance $\sqrt{q} \cdot 2^{d + 1} + 2^d$ of some $j'$ with $\beta_{j'} < \beta_j$. Hence, we get the same contradiction as before:
    \begin{align*}
        \beta_j &< \delta_i - (2\sqrt{q} + 1) \cdot 2^d \\
        &\leq \beta_j + \sqrt{q} \cdot 2^{d + 1} + 2^d - (2\sqrt{q} + 1) \cdot 2^d \\
        &= \beta_j.
    \end{align*}
    
    So each segment has a $\beta=0$ FC, as desired. Hence, similarly to the previous proof, we get for every demand $\delta_i \leq \sqrt{q} \cdot 2^{d + 1} + 2^d = (4\sqrt{q} + 2) \cdot 2^{d-1} \leq (4\sqrt{q} + 2) \cdot \ell_{ij}$. In particular, we get a $(4\sqrt{q} + 2)$-approximation using $(2 + \ceil{\sqrt{q}})^q \log \rho$ regions.  
\end{proof}

\end{appendices}

\end{document}

%% file: preamble.tex
\usepackage[shortlabels]{enumitem}

\usepackage{amsmath} 
\usepackage{amssymb} 
\usepackage{mathrsfs} 
\usepackage{mathtools} 
\usepackage{amsthm}
\usepackage{thmtools}
\usepackage{dsfont} 

\usepackage{appendix}
\usepackage{fullpage}
\usepackage{caption}
\usepackage{subcaption}
\usepackage{graphicx}
\usepackage[table,xcdraw]{xcolor}
\usepackage{wrapfig}
\usepackage{floatrow}
\usepackage{multirow}
\usepackage{makecell}
\usepackage{nicefrac}
\usepackage[normalem]{ulem} 
\usepackage{pdfpages}

\usepackage[ruled,vlined,linesnumbered]{algorithm2e}
\definecolor{forestgreen}{rgb}{0.13, 0.55, 0.13}

\SetCommentSty{mycommfont}
\SetKwInOut{Input}{Input}
\SetKwInOut{Output}{Output}
\DontPrintSemicolon

\definecolor{ForestGreen}{rgb}{0.0333,0.4451,0.0333}
\definecolor{DarkRed}{rgb}{0.65,0,0}
\definecolor{Red}{rgb}{1,0,0}
\usepackage[linktocpage=true,
pagebackref=true,colorlinks,
linkcolor=DarkRed,citecolor=ForestGreen,
bookmarks,bookmarksopen,bookmarksnumbered]{hyperref}

\usepackage{cleveref}
\usepackage{thm-restate}

\usepackage{tcolorbox}
\newtcolorbox{important}{arc=2mm, colback=DarkRed!3!white, colframe=DarkRed!75!black}

\usepackage{url}

\newtheorem{theorem}{Theorem}[section]
\newtheorem{lemma}[theorem]{Lemma}
\newtheorem{claim}[theorem]{Claim}

\newtheorem{corollary}[theorem]{Corollary}

\newtheorem{definition}[theorem]{Definition}

\newcommand*\ms[1]{\mathscr{#1}}

\newcommand*\mcal[1]{\mathcal{#1}}

\DeclarePairedDelimiter\abs{\lvert}{\rvert}
\DeclarePairedDelimiter\norm{\lVert}{\rVert}%
\DeclarePairedDelimiter{\ceil}{\lceil}{\rceil}

\makeatletter
\let\oldabs\abs
\def\abs{\@ifstar{\oldabs}{\oldabs*}}
\let\oldnorm\norm
\def\norm{\@ifstar{\oldnorm}{\oldnorm*}}
\let\oldceil\ceil
\def\ceil{\@ifstar{\oldceil}{\oldceil*}}
\makeatother

\newcommand{\eps}{\varepsilon}

\DeclareMathOperator{\cost}{cost}

\newcommand{\OPT}{\text{OPT}}

\DeclareMathOperator*{\argmin}{arg\,min}

\newcommand{\lpobjective}[2] { #1 \quad &\mathrlap{#2}}
\newcommand{\stconstraint}[4] { \text{s.t.} \quad & \quad & #1 & #2 #3 & \qquad & #4}
\newcommand{\lpconstraint}[4] { & & #1 & #2 #3  && #4}

\newcommand{\D}{\mcal{D}}
\newcommand{\F}{\mcal{F}}
\newcommand{\Dem}{D}

\makeatletter
\newcommand{\leqnomode}{\tagsleft@true}
\newcommand{\reqnomode}{\tagsleft@false}
\makeatother

\newif\ifcomments
\commentstrue
\commentsfalse 

\ifcomments
\usepackage[colorinlistoftodos,prependcaption,textsize=tiny]{todonotes}

\definecolor{dnotecol}{rgb}{0.20, 0.50, 0.80}

\else 

\fi 